\documentclass[conference]{IEEEtran}
\usepackage[ top=0.75in, bottom=1.02in, left=0.62in, right=0.62in]{geometry}
\renewcommand{\thesubsection}{\alph{subsection}}

\usepackage{cite,url}
\usepackage{graphicx}

\usepackage{graphicx}
\usepackage{sidecap}
\usepackage[caption=false,font=footnotesize,labelformat=empty]{subfig}

\usepackage[cmex10]{amsmath}
\usepackage{amssymb}

\usepackage{array}
\newtheorem{theorem}{Theorem}

\newenvironment{proof}[1][Proof]{\begin{trivlist}
\item[\hskip \labelsep {\bfseries #1}]}{\end{trivlist}}

\usepackage{url}
\usepackage{fontenc}
\usepackage[utf8]{inputenc}
\usepackage{import}
\usepackage{mathtools}
\usepackage{tabularx}
\usepackage{siunitx}

\newsavebox{\ieeealgbox}

\newcolumntype{M}[1]{>{\arraybackslash}m{#1}}
\newcolumntype{N}{@{}m{0pt}@{}}
\begin{document}

\title{Impact of End-User Behavior on User/Network Association in HetNets \vspace{-0.8cm} }
\author{\IEEEauthorblockN{Mohammad~Yousefvand \IEEEauthorrefmark{1},
Mohammad~Hajimirsadeghi\IEEEauthorrefmark{1},
and Narayan B. Mandayam\IEEEauthorrefmark{1}}
\IEEEauthorblockA{\IEEEauthorrefmark{1}WINLAB, Rutgers, The State University of New Jersey, North Brunswick, NJ, USA.\\ Email: \{my342, mohammad, narayan\}@winlab.rutgers.edu}

\thanks{This work is supported in part by the U.S. National Science Foundation under Grant No. 1421961 and Grant No. ACI-1541069}

}
%
\IEEEoverridecommandlockouts

\maketitle

\begin{abstract}
We study the impact of end-user behavior on user/network association in a HetNet with multiple service providers (SPs). Specifically, we consider the uncertainty in the service guarantees offered by SPs in a HetNet, and use Prospect Theory (PT) to model end-user decision making. We formulate user association with SPs as a multiple leader Stackelberg game where each SP offers a data rate to each user with a certain service guarantee and at a certain price, while the user chooses the best offer among multiple such bids. Using the specific example of a HetNet with one cellular base station and one WiFi access point, we show that when the end users underestimate the advertised service guarantees, then some of the Nash Equilibrium strategies under the Expected Utility Theory (EUT) model become infeasible under PT, and for those Nash Equilibria that are feasible under both EUT and PT, the resulting user utilities are less under PT. We propose resource allocation and bidding mechanisms for the SPs to mitigate these effects.

\end{abstract}


\section{Introduction}
\label{Intro}

The emergence of HetNets for network densification in future wireless networks has lead to extensive user association studies in this context \cite{Ramazanali:2016:SUA, Liu:2016:UAI}. The range of approaches include evolutionary game theory \cite{Han:2014:EFU, Yousefvand:2015:MNC}, auction based models \cite{Yousefvand:2017:DES}, matching theory \cite {Han:2016:MBJ}, Stackelberg, the competitive approach using Colonel Blotto game\cite{Mohammad:2017:DCB,Mohammad:2017:IND}, and other gametheoretic models \cite{Haddad:2013:GTA, Sokun:2015:QGU, Zhou:2015:DUA, Lin:2015:OUA, Chen:2014:NBS}. However, most of these mechanisms essentially are borne out of expected utility theory (EUT) based approaches. When a service provider (SP) controls access to end-users via differentiated and hierarchical monetary pricing, then the performance of the network is directly subject to end-user decision-making that has shown to deviate from EUT in many cases \cite{Kahneman:1979:PT, Li:2014:WUI, Yang:2014:IEU, Yang:2015:PPC}. 
In this work, we use Prospect Theory \cite{Kahneman:1979:PT}, a Nobel prize winning theory that explains real-life decision-making and its deviations from EUT behavior, to study user decisions in wireless HetNets. To do so, we first formulate the user association problem in HetNets as a Stackelberg game between SPs and user, in which WiFi and cellular SPs as the leaders of the game make offers to the user, and the user as a follower makes a decision about the received offers. Then based on the user response to received offers, the SPs optimize their bids to maximize their utility. By considering a convex pricing function for the SPs and a concave payoff function for the users, we compare the utility of both user and SPs under EUT and PT. We derive all possible pure strategy and mixed strategy NEs for the proposed Stackelberg game. We also provide the conditions under which the existence of such NEs are guaranteed, for both EUT and PT cases. To see the effects of heterogeneity of the SPs, we also compare the results for both symmetric and non-symmetric SP models. To the best of our knowledge, this is the first paper which address PT effects on user association in HetNets.

The rest of this paper is organized as follows. In section \ref{sec:StackelbergGame}, after describing the HetNet model, we introduce the model of interactions between SPs and users as a Stackelberg game, and formulate user association problem. In section \ref{sec:BRS}, we discuss the best response strategies for all players, and in section \ref{sec:NEEUT}, we derive all the possible NE strategies, and the conditions under which the existence of such NEs can be guaranteed under EUT. In section \ref{sec:NEPT} we show the impacts of PT on user association and consequently its effect on SP utilities. In section \ref{sec:PBS} we introduce resource allocation and bidding strategies for SPs to compensate part of their lost utility under PT. We validate the analysis in the paper using simulation results in section \ref{sec:SimulationResults} and conclude in section \ref{sec:conclusion}.

\section{System Model and Problem Formulation}
\label{sec:StackelbergGame}

\subsection{Network Model}
\label{NM}
To study user association in HetNets, we developed a two-tier HetNet scenario which includes $N$ wireless users that are randomly distributed within the coverage area of $K$ base stations. As shown in Fig. \ref{Fig:figl}, in our HetNet model there is one macrocell LTE BS located in the center of the area and $K-1$ overlaid small cell WiFi access points who are competing with each other to serve the users in the HetNet. We assume each user in the HetNet receives several offers from service providers (SPs) in both cellular and WiFi tiers, and upon receiving such offers, the user makes a decision to accept or reject any of the received offers. As we can see in Fig. \ref{Fig:figl}, the number of users (load) in each small cell could be different than other cells, and also the number of covering BSs for each user can differ from other users. Moreover, depending on the level of noise and attenuation on each user-BS link, the users in each part of the network may experience different channel gains than other users.

\begin{figure}[tb!]
  \includegraphics[width=\linewidth]{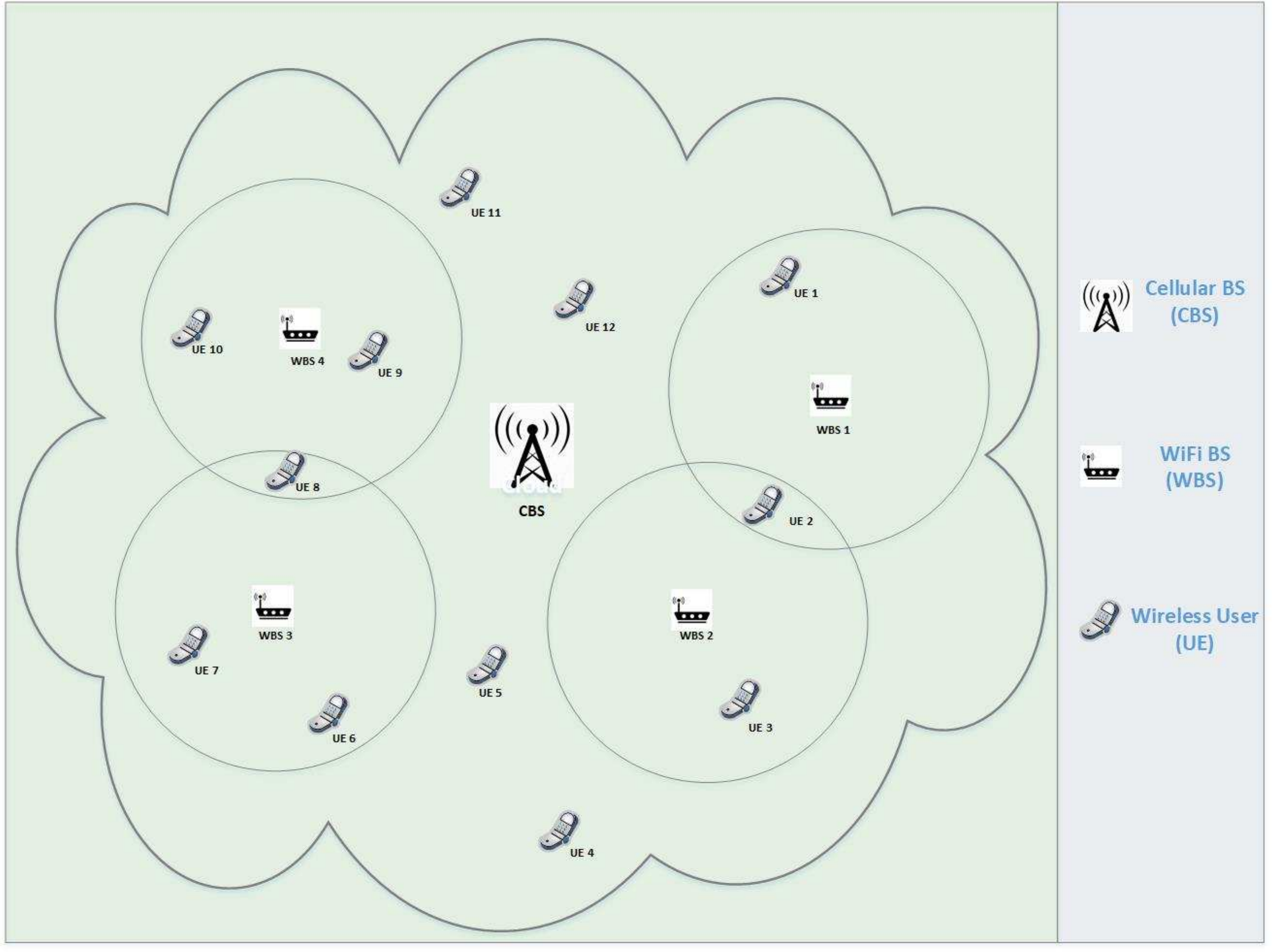}
  \centering
  \caption{Heterogeneous Network (HetNet) Model.}
  \label{Fig:figl}
\end{figure}

\subsection{Stackelberg Game for User Association in HetNets}
\label{subsec:SM}
In our HetNet model, each user receives $K$ different offers from all $K$ base stations, i.e., one offer from the cellular BS and $K-1$ offers from the WiFi BSs. To enable multihoming, we assume each user can be simultaneously connected to both cellular and WiFi SPs to receive the data service. Specifically, we assume that each user can only be associated to the cellular SP and the best serving WiFi SP for that specific user, which is the SP who offers a bid with highest utility among all WiFi SPs. Hence, we model user association problem in this work as a Stackelberg game with two leaders and one follower, where SPs act as the leaders who make service offers to the user, and the user serves as a follower who accepts or rejects the received offers. Since we have $N$ users in our HetNet model, to solve the user/network association problem in a distributed manner, we need to solve $N$ stackelberg games each with three players including WiFi SP and cellular SP as the leaders and one user as the follower. Since these $N$ Stackelberg games are independent, assuming the BS's maximum bandwidth budget per user is fixed and the same for all users as shown in Eq. \ref{eq14}, we can consider and solve only one of these Stackelberg games without loss of generality. Note that finding the optimal bandwidth/power allocation for SPs is a NP-hard problem \cite{Yousefvand:2017:DES, Hajisami:2017:DJP} and not the focus of this paper. In the remainder of this paper, we focus on one of these games. 
Using the index $w$ for user's preferred WiFi SP and the index $c$ for cellular SP, we denote the $bids$ of WiFi and cellular SPs with triples $(b_w, r_{EUT}(b_w), BW_{w,EUT} )$, and $(b_c, r_{EUT}(b_c), BW_{c,EUT} )$, respectively, in which the first term shows the advertised data rate, the second term is the proposed price for the offered data rate, and the third term is the amount of BW that will be allocated to the user by each SP. User decisions are binary, which means the user either accepts a bid or rejects it, and there is no probabilistic decision by user. We denote user decisions on cellular and WiFi SPs offers with $p_c$ and $p_w$ respectively, which are binary variables. So, the tuple $(p_c,p_w)$ represents user’s strategy with regard to the received bids, hence, the user has four possible strategies $(0,0)$, $(0,1)$, $(1,0)$ and $(1,1)$. All three players have a cost function and a benefit function, and their utility functions are simply the difference of the cost and benefit functions. We represent user’s utility function as:
\begin{equation}
\label{eq1}
U_{user,{\rm{ }}EUT}{\rm{ }}\left\{ {p_c,p_w{\rm{ }}} \right\} = H\left( {B_{joint}{\rm{ }}} \right) - c_{user}{\rm{ }}\left( {p_w,p_c{\rm{ }}} \right)
\end{equation}
where $H(B_{joint})= \delta(B_{joint})^{1/\theta} $ is the users benefit function which is a logarithmic concave function of the aggregate data rate achieved, $B_{joint}$, by the user, and $c_{user}{\rm{ }}\left( {p_w,p_c{\rm{ }}} \right)= p_w r_{EUT}(b_w) + p_c r_{EUT} (b_c)$ is the user’s cost function which shows the aggregate price that must be paid by user to the SPs for each $(p_c, p_w)$ strategy. The aggregate data rate is defined as $ B_{joint} = b_c \bar {F} _{B_c}( b_c,BW_{c,EUT})p_c + b_w \bar {F}_{B_w}( b_w,BW_{w,EUT})p_w $ where $\bar F_{B_c}( b_c,BW_{c,EUT})$ is the service guarantee of the offer received from the cellular SP, and denotes the probability of having the actual data rate of user from cellular SP, $B_c$, equal or higher than the advertised data rate by cellular SP, $b_c$. Similar definition holds for $\bar F_{B_w}( b_w,BW_{w,EUT})$ which is the service guarantee of the offer received from WiFi SP.

Once the user chooses its best response strategy, $(p_c^*, p_w^*)$, the SPs will respond with their best response strategies to maximize their own utilities based on the received user decision. The utility of the WiFi SP, $U_{SP,w}$ is defined as
\begin{equation}
\label{eq2}
U_{SP,w} = p_w r_{EUT}( {b_w}) - C_w( b_w,BW_{w,EUT}),
\end{equation}
and the utility of cellular SP, $U_{SP,c}$ is defined as
\begin{equation}
\label{eq3}
U_{SP,c} = p_c r_{EUT}( {b_c}) - C_c( b_c,BW_{c,EUT}),
\end{equation} 
where, the first term in both of these equations is the SPs' expected payoff from the user, and the second term is their incurred service cost. The SPs' payoff from the user is equal to the offered price in their bids if the user accepts their offers, otherwise their payoff from the user is equal to zero. Overpricing may lead the user to reject their bids, and underpricing may lead their utility to be negative. So, choosing a proper pricing function is critical for the SPs. In this work, we assume both SPs use convex pricing functions, as $r_{EUT}(b_w)= {\alpha_1} (b_w)^{\beta_1}$, and $r_{EUT}(b_c)= {\alpha_2} (b_c)^{\beta_2}$, $({\beta_1},{\beta_2}>1)$, where ${\alpha_1}$ and ${\beta_1}$ are payoff parameters for the WiFi SP, and ${\alpha_2}$ and ${\beta_2}$ are payoff parameters for the cellular SP. We also assume the SPs use linear cost functions, as $C_w( b_w,BW_{w,EUT})= c_1 (b_w) + c_2(BW_{w,EUT})$, and $C_c( b_c,BW_{c,EUT})= c_3 (b_c) + c_4(BW_{c,EUT})$, where $c_1$ and $c_2$ are cost coefficients for the WiFi SP, and  $c_3$ and $c_4$ are cost coefficients for the cellular SP. To satisfy the user's minimum data rate constraint, the SPs must ensure that their offered data rate is higher than the minimum data rate required by the user, $b_{min}$. Thus, the data rate constraints for the WiFi and the cellular SPs will be defined as below, respectively:
\begin{align}
\label{eq4}
b_w \bar F _{B_w}( {b_w,BW_{w,EUT}}) \ge b_{min},\\
\label{eq5}
b_c \bar F _{B_c}( {b_c,BW_{c,EUT}}) \ge b_{min}.
\end{align}

\section{Best Response Strategies for players}
\label{sec:BRS}

\subsection{User Optimization Problem}
Upon receiving the offers from the SPs, the user will run an optimization problem to find its best strategy with regard to received offers. We assume the user's payoff function from the received data is a concave function, as defined in Eq. \ref{eq1}, in which the user's utility is not linearly increased with increasing the data rate. It means that as long as the minimum data rate constraint is satisfied, the user is not willing to pay extra price with linear relation to the extra data rate offered by SPs. To find its best response strategy,$(p_c^*, p_w^*)$, user will run the following optimization problem (denoted as $Max1$): \\\\
\textbf {Max1 Problem:} \textit{User's Utility Maximization}.\\
----------------------------------------------------------------------------
\begin{align}
\label{eq6}
& \max_{p_c, p_w}~[\,\delta(B_{joint})^{1/\theta} -~p_w{\alpha_1} (b_w)^{\beta _1}-~p_c{\alpha_2} (b_c)^{\beta _2}~] && \\
& \text{subject to}\nonumber \\
\label{eq7}
& B_{joint}\ge~b_{min}  \\
\label{eq8}
&  \delta(B_{joint})^{1/\theta} \ge~~p_w{\alpha_1} (b_w)^{\beta _1}+~p_c{\alpha_2} (b_c)^{\beta _2}, \\
\label{eq9}
& p_c, p_w \in \{ 0,1\}
\end{align}

As shown above, the user has two major constraints for bid selection. The first constraint, shown in Eq. \ref{eq7}, is the user's data rate constraint which ensures the expected data rate for the user is higher than its minimum required data rate, $b_{min}$. The second constraint defined in Eq. \ref{eq8} is the user's utility constraint which guarantees a positive utility for the user from its strategy. 

\subsection{SPs Optimization Problems}
When the SPs receive user's decision with regard to their offers, they choose their best response strategy. The the best response strategy $(b_w^*, BW_{w,EUT}^*)$ for the WiFi BSs is obtained by solving the optimization problem below (denoted as $Max2$):\\\\
\textbf {Max2 Problem:} \textit{WiFi SP's Utility Maximization}.\\
----------------------------------------------------------------------------
\begin{align}
\label{eq10}
& \max_{b_w, BW_{w,EUT}}~[p_w {\alpha_1} (b_w)^{\beta _1} - (c_1 b_w + c_2 BW_{w,EUT})~] \\
& \text{subject to}\nonumber \\
\label{eq11}
& 0 \leq BW_{w,EUT} \leq BW_{w,max}, \\
\label{eq12}
& 0 \leq b_w \leq b_{w,max}, \\
\label{eq13}
&  b_w \bar F _{B_w}( {b_w,BW_{w,EUT}}) \ge b_{min},
\end{align}
in which, $BW_{w,max}$ is the maximum amount of bandwidth that can be allocated to the user by the WiFi SP, and $b_{w,max}$ is the maximum achievable data rate by the user from the WiFi SP considering $BW_{w,max}$ and the gain of the channel between the WiFi SP and the user. We assume the SPs use a proportionally fair bandwidth allocation algorithm to determine the maximum amount of bandwidth that can be allocated to each of their users. Assuming $BW$ as the total amount of bandwidth available at SP $i$, $i \in\{w,c\}$, and $N$ as the total number of users in our HetNet, the maximum amount of bandwidth that can be allocated to each user by the SP $i$, $BW_{i,max}$ is given as

\begin{equation}
\begin{aligned}
\label{eq14}
& BW_{i,max}= (G_{BA}*~BW)/\sum_{j=1}^{N}({a_j * c_j}),\\
\end{aligned}
\end{equation}
in which $G_{BA}$ is the gain of bandwidth allocation which is less than one due to the guard bands between channels for preventing interference, $a_j$ is a binary variable showing the activity of the users which $a_j=1$ if the user $j$ is active and has data demand, otherwise $a_j=0$,  and $c_j$ is also a binary variable representing the coverage of the user $j$ by the the BS $i$, where $c_j=1$ if the SINR of the link between the user $j$ and the BS $i$ is higher than a certain threshold to have a reliable transmission, otherwise $c_j=0$. Considering $BW_{i,max}$, the maximum achievable data rate by user, $b_{i,max}$ is given by:

\begin{equation}
\begin{aligned}
\label{eq15}
& b_{i,max}= BW_{i,max} \log (1+\frac{P h_j a_j c_j}{\sigma^2}),\\
\end{aligned}
\end{equation}
where, $P$ is the transmit power of the BS, $h_j$ is the channel gain between the user $j$ and the SP $i$, and $\sigma^2$ is the noise variance over the transmission channel. Similarly, the cellular SP runs $Max3$ optimization problem, which is defined exactly similar to $Max2$, except the index $w$ is replaced with the index $c$ and the parameters $\alpha_1$, $\beta _1$, $c_1$, $c_2$ are replaced with the parameters $\alpha_2$, $\beta _2$, $c_3$, $c_4$, respectively. In Theorem \ref{Theorem1}, we prove that for both WiFi and cellular SPs, the best response strategies derived from $Max2$ and $Max3$ optimization problems, satisfy the minimum data rate constraint with equality. \\

\begin{theorem}
\label{Theorem1}
The SPs best response strategies, derived from $Max2$ and $Max3$ problems, will always satisfy the minimum data rate constraint in the boundary of its feasibility region, i.e. we always have $b_w^* \bar F _{B_w}( {b_w^*,BW_{w,EUT}^*}) = b_{min}$, and $b_c^* \bar F _{B_c}( {b_c^*,BW_{c,EUT}^*}) = b_{min}$ for the WiFi and the Cellular SPs, respectively.
\end{theorem}

\begin{proof}
We prove this by contradiction for the WiFi SP.
Assume $({b_w^*,BW_{w,EUT}^*})$ is the optimal solution for $Max2$ problem, and $BW_{w,EUT}^*$ is not a marginal BW, i.e. $b_w^* \bar F _{B_w}( {b_w^*,BW_{w,EUT}^*}) \neq b_{min}$. Considering Eq. \ref{eq13}, we can infer that $b_w^* \bar F _{B_w}( {b_w^*,BW_{w,EUT}^*}) > b_{min}$ (1).
In this case $\exists~BW_{w,EUT}^\prime $ such that $b_w^* \bar F _{B_w}( {b_w^*,BW_{w,EUT}^\prime}) = b_{min}$ (2). From (1) and (2), and considering the direct relation between $F _{B_w}( {b_w^*,BW_{w,EUT}^*})$ and $BW_{w,EUT}^*$, we can infer that $BW_{w,EUT}^\prime<BW_{w,EUT}^*$. Now, considering the WiFi SP's utility function defined in Eq. (11) and due to its reverse relation with the advertised bandwidth, $BW_{w,EUT}$, we can infer that $U_{SP,w}({b_w^*}, BW_{w,EUT}^\prime)>U_{SP,w}({b_w^*}, BW_{w,EUT}^*)$ which is in conflict with the initial assumption that $({b_w^*,BW_{w,EUT}^*})$ is the optimal solution for $Max2$ problem. So, the proof is complete. The same proof is valid for the cellular SP.
\end{proof}

\section{NE Existence Analysis under EUT}
\label{sec:NEEUT}

We derive all the potential Nash Equilibrium strategies for the user in the proposed Stackelberg game under EUT. To see the effects of the SPs heterogeneity on the existence of NE, we consider both symmetric and asymmetric SPs cases.
\subsection{Symmetric SPs under EUT}
\label{SymmetricEUT}
Under the symmetric model, we assume both SPs offer the same data rate and use the same pricing and cost functions, i.e. $\alpha_1=\alpha_2=\alpha$, $\beta _1=\beta _2=\beta$, $c_1=c_3=c_b$, $c_2=c_4=c_{BW}$. In this situation, both SPs offer $(b^*, BW^*)$ with the price of $r_{EUT}(b^*)=\alpha(b^*)^\beta$, and this offer costs $c_b b^* +~c_{BW} {BW}^*$ for them. Table \ref{fig:Table1} summarizes the NE strategies for the user and their existence conditions under EUT for both symmetric and asymmetric SPs cases. As we can see from this table, under symmetric case user has two pure strategy NEs, $(0,0)$ and $(1,1)$, and one mixed strategy NE, $(1,0)/(0,1)$. In fact, if the user's payoff coefficient, $\delta$, is less than a threshold which guarantees positive utility for the user, it will choose $(0,0)$ strategy and rejects both of received offers. However, if the user's payoff coefficient is big enough such that the expected payoff of the user from expected data rate is higher than the offered price for it, i.e. $\delta(b_{min})^{1/\theta} \geq \alpha(b^*)^\beta $ then the user will accept at least one of the SPs offers. Note that in this situation, the user's expected data rate from both SPs is $2b_{min}$ and the expected price for it is $2\alpha(b^*)^\beta$. However, due to the concavity of the user's payoff function, when the data rate is doubled, the resulting payoff from that data rate will not be doubled for the user, and is less than twice of the initial payoff. In this situation, if the extra payoff from the second SP, $H(2b_{min})-H(b_{min})=\delta(2^{1/\beta}-1)(b_{min})^{1/\theta}$, is higher than the extra price asked for it by second SP, $\alpha(b^*)^\beta$, then the user will choose $(1,1)$ strategy and accepts both of received offers. Otherwise, the user accepts one of the received offers only, and due to the symmetry of received bids, the user will choose one of WiFi or cellular offers with the same chance, each time. It leads the user to have a mixed strategy NE $(0,1)/(1,0)$ which means in half of the times the user will choose the WiFi SP offer, and in the other half the user will choose the cellular SP offer. Consequently, because the SPs know that their offers are not going to be accepted by the user in 50\% of times, they will also choose a mixed strategy of $(b^*, BW^*)/(0,0)$ which means in half of the times they prefer to stay silent and not offer any data rate to prevent negative utility as a result of being rejected by the user.

\begin{table}
\setlength{\belowcaptionskip}{-5pt}
  \caption{NE table for user under EUT.}
 \includegraphics[width=11cm, height=8cm]{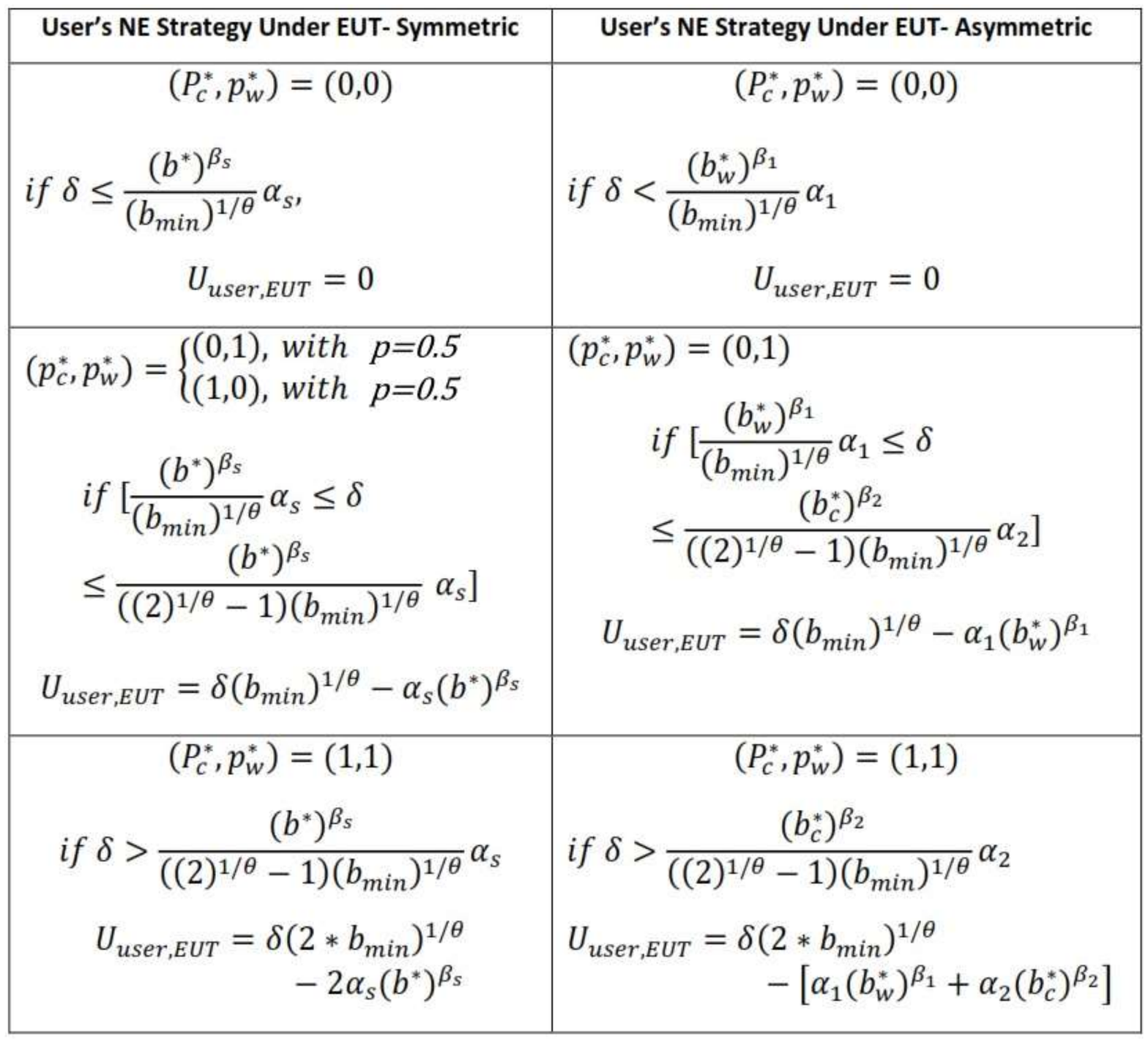}
  \centering
  \label{fig:Table1}
  \vspace{-3em}
\end{table}

 \subsection{Asymmetric SPs under EUT}
 \label{AsymmetricEUT}
 Under the asymmetric model, we assume the SPs have different pricing and cost functions, as defined in subsection \ref{sec:StackelbergGame}.\MakeUppercase{\ref{subsec:SM}}. We denote the WiFi and the cellular SPs best response strategies under asymmetric model as $(b_w^*, BW_{w,{EUT}}^*)$, and $(b_c^*, BW_{c,{EUT}}^*)$, respectively. Here we assume due to better channel conditions and SINR, the WiFi SP offers its data service with a cheaper price as compared to the cellular SP, i.e. ${\alpha_1} (b_w^*)^{\beta _1} \leq {\alpha_2} (b_c^*)^{\beta _2}$. In fact, because the WiFi SP requires less amount of BW to satisfy user minimum data rate, it's cost in serving user is lower than the cellular SP, thus it justifies the WiFi SP's demand for the lower price. Note that, despite offering different bids by the WiFi and the cellular SPs, because of Theorem 1, still the user's expected data rate from  each of them is equal to $b_{min}$ considering the service guarantees, i.e. $b_w^* \bar F _{B_w}( {b_w^*,BW_{w,EUT}^*})= b_{min}$, and $b_c^* \bar F _{B_c}( {b_c^*,BW_{c,EUT}^*})=b_{min}$. As shown in Table \ref{fig:Table1}, we don't have any mixed strategy NE under asymmetric model. The reason is that, when the user's best response strategy is to accept only one of the received bids, it will always accept the bid from superior SP which leads to a higher utility for him, which is the WiFi SP in our case. That's why the user does not have $(1,0)$ NE strategy under asymmetric model, and it only has three pure strategy NEs, including $(0,0)$, $(0,1)$ and $(1,1)$. According to the user's NE table under EUT, if the user's payoff coefficient, $\delta$, is not bigger than a threshold to guarantee a positive utility for the user from accepting any bid, then the user will choose $(0,0)$ strategy and rejects both of the received offers. When the user's expected utility from the WiFi SPs offer is positive but accepting both offers together reduces the user's utility, the user will choose (0,1) strategy to accept the WiFi SP's bid and reject the cellular SP's offer. However, when the user's payoff coefficient is bigger than a threshold shown in Table \ref{fig:Table1}, then the extra payoff of the user achieved from the cellular SP is higher than the price asked by the cellular SP, and it convinces the user to accept both offers even though the WiFi SP offer is enough to guarantee its minimum data rate. In this situation, the user's expected utility from $(1,1)$ strategy is higher than its expected utility from $(0,1)$ strategy, and thus the user chooses $(1,1)$ to achieve higher utility.

 \section{NE Existence Analysis under PT}
\label{sec:NEPT}

So far, our analysis on the existence of NE and derivation of the user's best response strategies was based on EUT. In fact, when the user is making decisions about a system with some uncertainty in system parameters, like the service guarantees in our model, EUT fails to describe the user decisions precisely. In this situation, we use Prospect Theory to model user decisions and to capture the psychophysics of end-user decision making \cite{Yang:2015:PPC}. To do so, we assume the SPs still make the same offers to the user as they did under EUT, however, we assume the user makes decisions about the received offers based on PT. In this work, we just focus on the probability weighting effect (PWE) of PT to see its effects on the NEs of our Stackelberg game. We use the Prelec function \cite{Prelec:1998:PWF} to model the PWE under PT:
\begin{align}
\label{eq20}
&& w(p)=exp(-{(-ln(p)}^\alpha ),~~(0 < \alpha< 1 ).
\end{align}
 This function is a regressive and s-shaped function which is concave in $0<p<1/e$ region and convex in $1/e<p<1$ region, and $w(p)>p$ in the former domain while $w(p)<p$ in the later. Considering this function to model PWE of PT, we can infer that the user overestimates the service guarantees of the received offers if the advertised service guarantees are less than $1/e=0.37\%$, and user over estimates them if advertised service guarantees are higher than $0.37\%$. In this work, assuming that SP networks are well designed to offer service guarantees higher than $1/e$, we focus on the under estimating of service guarantees by the user under PT. It is also justified by the fact that end-users in real world wireless networks typically perceive the quality of their service as lower than that advertised by the SPs.
\subsection{Symmetric SPs under PT}
In this case, we assume both SPs use the same pricing and cost functions as described in subsection \ref{sec:NEEUT}.\MakeUppercase{\ref{SymmetricEUT}} for the EUT case. However, the user makes decision about the received offers based on PT, and thus user's perception of the service guarantee will be affected by PWE of PT. In fact, the user under estimates the service guarantees of the received offers under PT, assuming service guarantees to be higher than $1/e$. Table \ref{fig:Table2} summarizes the potential NEs for the user in the Stackelberg game under PT for both symmetric and asymmetric SPs cases. Under PT we only have two pure strategies, in which the user either accepts both of the received offers or rejects both of them, and we don't have any NE strategy for the user in a form of $(0,1)$,or $(1,0)$ under which the user receives its service from only one SP. In fact, the most important difference between PT and EUT in our model is that, under PT, none of the individual SPs can satisfy user constraints, independently. The reason for this is that based on Theorem 1, the SPs always satisfy the data rate constraint with equality. Hence, under PT any under estimation of service guarantees by the user will result in violation of the data rate bid selection constraint for the user. In this situation, the expected data rate for the user from any individual SP is less than $b_{min}$, as $w(\bar F _{B_w}( {b_w^*,BW_{w,EUT}^*}))<\bar F _{B_w}( {b_w^*,BW_{w,EUT}^*})$, and $w(\bar F _{B_c}( {b_c^*,BW_{c,EUT}^*}))<\bar F _{B_c}( {b_c^*,BW_{c,EUT}^*})$. Hence, the user cannot accept any individual SP offer, and both $(0,1)$ and $(1,0)$ strategies are not feasible for the user under PT. Unlike in the EUT case, because the data rate constraint is not guaranteed to be satisfied for the user under PT, the user has two conditions for each NE strategy, to ensure that both utility constraint and data rate constraint are satisfied. If both of these constraints are satisfied, then the user chooses $(1,1)$ strategy and accepts both offers, otherwise, the user rejects both offers by choosing $(0,0)$ strategy.
\begin{table}
\setlength{\belowcaptionskip}{-5pt}
\caption{NE table for user under PT.}
  \includegraphics[width=9.5cm,height=8cm]{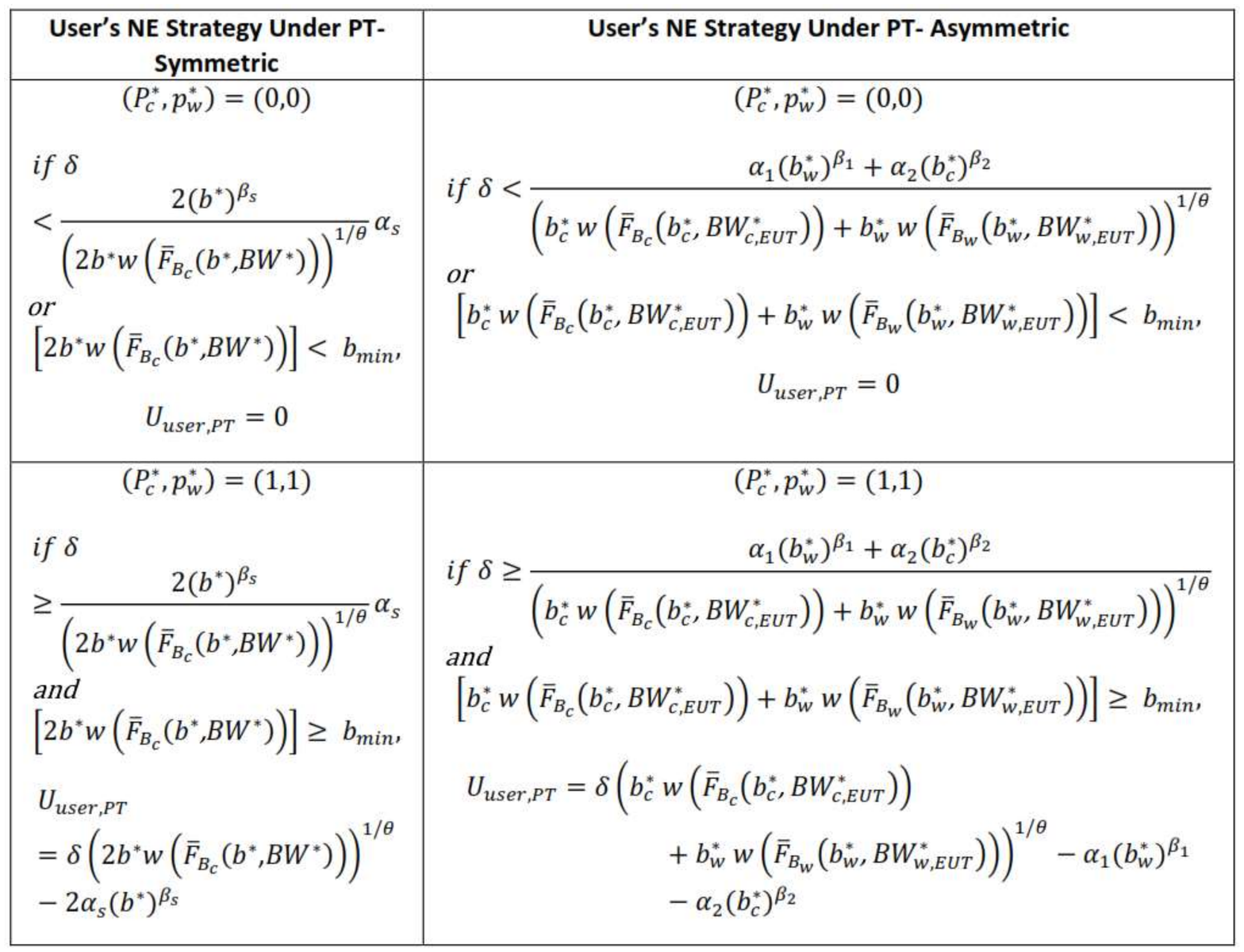}
  \centering
  \label{fig:Table2}
   \vspace{-3em}
\end{table}
 
 \subsection{Asymmetric SPs under PT}
In this case as described in subsection \ref{sec:NEEUT}.\MakeUppercase{\ref{AsymmetricEUT}}, we assume the SPs use different cost and pricing functions. As mentioned before, we assume the user will underestimate the service guarantees of the SP offers under PT, due to PWE. Hence, user's estimation of the expected data rate, and the expected payoff from the received offers will change accordingly as both of them are functions of service guarantees. Potential NEs for the user in the Stackelberg game under PT with asymmetric SPs are shown in Table \ref{fig:Table2}. Due to underweighting of service guarantee by the user under PT, none of individual SPs can satisfy the users data rate constraint independently, hence, none of $(0,1)$ and $(1,0)$ strategies are feasible for the user. In this case, if the user's payoff from both SPs under $(1,1)$ strategy is positive, and also, the users expected data rate from both SPs is higher than its minimum data rate, the user chooses $(1,1)$ strategy and accepts both of the received offers. Otherwise, the user rejects both of received offers by adopting $(0,0)$ strategy. As we can see in Table \ref{fig:Table2}, the conditions for having $(1,1)$ NE strategy for user under PT are stricter than in EUT case, due to the underweighting of service guarantees which makes it more difficult for both data rate and utility constraints to be satisfied under PT. Further, considering the infeasibility of $(0,1)$ and $(1,0)$ strategies, and noting that the user has only four potential strategies, we can infer that the chance of choosing $(0,0)$ strategy by the user and rejecting both of SP offers under PT is much higher than it under EUT.

\section{Proposed Bidding Strategy for SPs under PT}
\label{sec:PBS}
In previous section, we inferred that under PT if the user underestimates the advertised service guarantees, it is more likely for the user to reject the SPs offers by choosing $(0,0)$ strategy. This can reduce the utility of SPs if they do not redesign their bidding strategies. So, in order to help the SPs to cope with user decisions under PT, and increase the chance of their offers to be accepted by the user, we propose a new bidding strategy for the SPs under PT based on bandwidth expansion \cite{Yang:2015:PPC}.

\subsection{Bandwidth Expansion Under PT}
As mentioned before, under PT user replaces its subjected service guarantee, for example $w(\bar F _{B_w}( {b_w^*,BW_{w,EUT}^*}))$ for the WiFi SP, with the service guarantee advertised by the SPs, $\bar F _{B_w}( {b_w^*,BW_{w,EUT}^*})$ for the WiFi. And if the advertised service guarantee is higher than $1/e$, the user will underestimate it, which means we have $w(\bar F _{B_w}( {b_w^*,BW_{w,EUT}^*}))<\bar F _{B_w}( {b_w^*,BW_{w,EUT}^*})$. In such a situation if the SPs offer the same bid as in EUT, the user will reject their bids as the data rate constraint of the user will not be satisfied by such offers under PT. So, for SPs to convince the user to accept their bids, one way is to increase their offered bandwidth such that despite the user's under estimation of the service guarantee, still the expected data rate for the user under PT is higher than its minimum data rate threshold, $b_{min}$. If we denote $BW_{w,PT}^*$, and $BW_{c,PT}^*$ as the amount of BW required by the WiFi and the cellular SPs, respectively to convince the user to accept their offers under PT, assuming their offers got accepted under EUT, we must have:
\begin{align}
& BW_{w,PT}^* \geq \bar F _{B_w}^{-1}( {\lambda_w,b_w^*}) \\
&\text{where}\nonumber \\
&\lambda_w=~w^{-1}(\bar F _{B_w}( {b_w^*,BW_{w,EUT}^*})
\end{align}
where, $( {b_w^*,BW_{w,EUT}^*})$ is the WiFi SP bid under EUT. In fact the WiFi SP must expand its offered bandwidth so as to offer $b_w^*$ data rate with the service guarantee of $\lambda_w$, where $w(\lambda_w)=~\bar F _{B_w}( {b_w^*,BW_{w,EUT}^*})$. This way, the extra bandwidth offered by the WiFi SP compensates the under estimation of the service guarantee by user under PT. Same conditions are valid for cellular SP.

\section{Simulation Results}
\label{sec:SimulationResults}
In this section, we provide several simulation results to validate the efficiency of our model. We consider a HetNet scenario similar to the one presented in Fig. \ref{Fig:figl}, in which there are $N$ randomly distributed users that are covered by 9 BSs. There is one cellular SP in the center of a macro cell and 8 overlaid WiFi SPs. We assume that the maximum coverage radius for WiFi SPs is 300 ft, while, all users are located within the coverage range of the cellular SP and can get served by it. We use Hata propagation model for urban environments to capture the effects of path loss on the user-BSs links. The table below contains the description and values of different parameters we have used in our simulations.
\begin{figure}[tb!]
  \includegraphics[width=\linewidth]{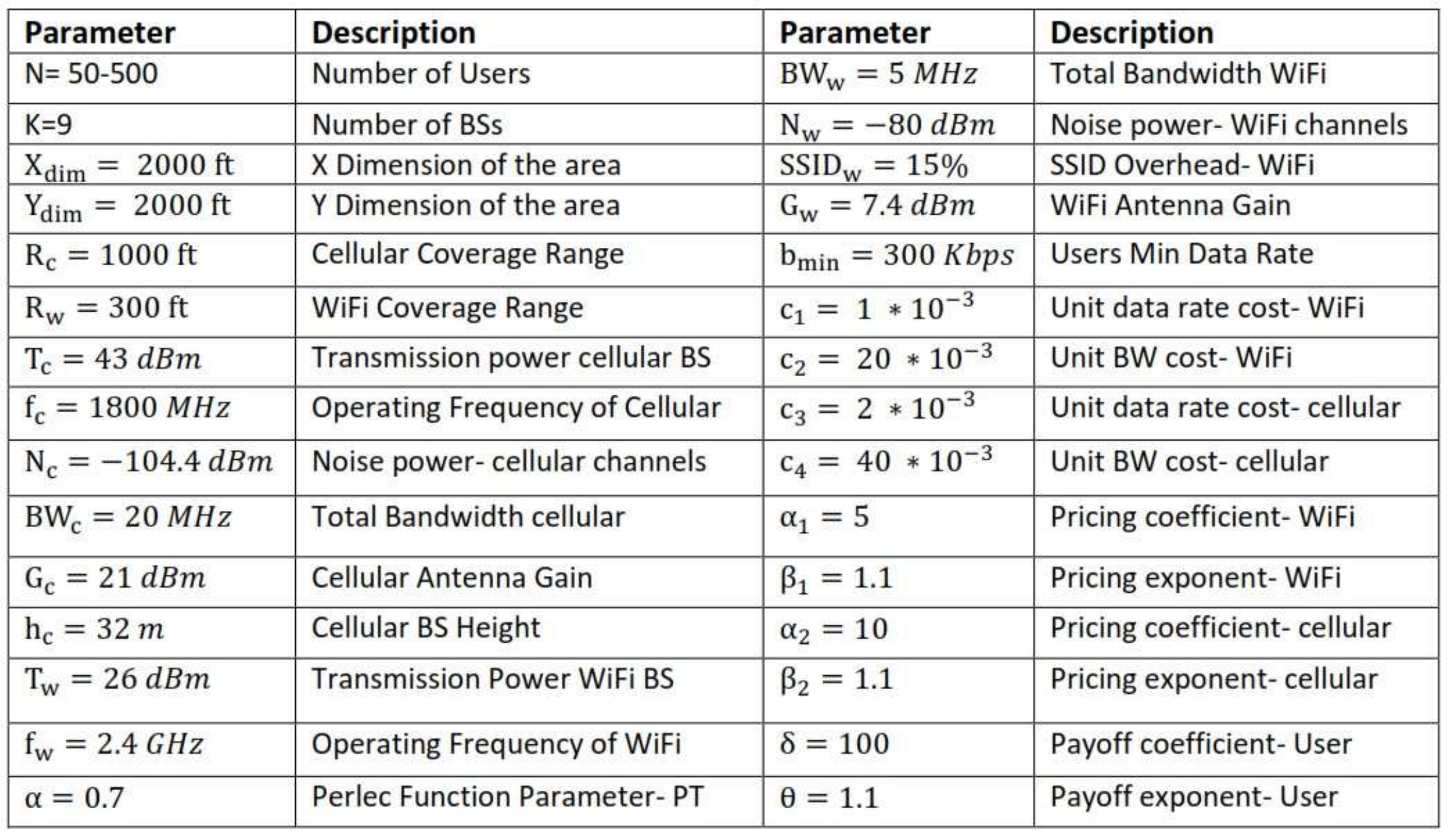}
  \centering
  \caption{Simulation Parameters.}
  \label{fig7}
\end{figure}

Fig. \ref{fig8} compares the sum utility of BSs under these three scenarios while we change the number of users (load) from $N=50$ to $N=500$.
\begin{figure}[tb!]
  \includegraphics[width=\linewidth]{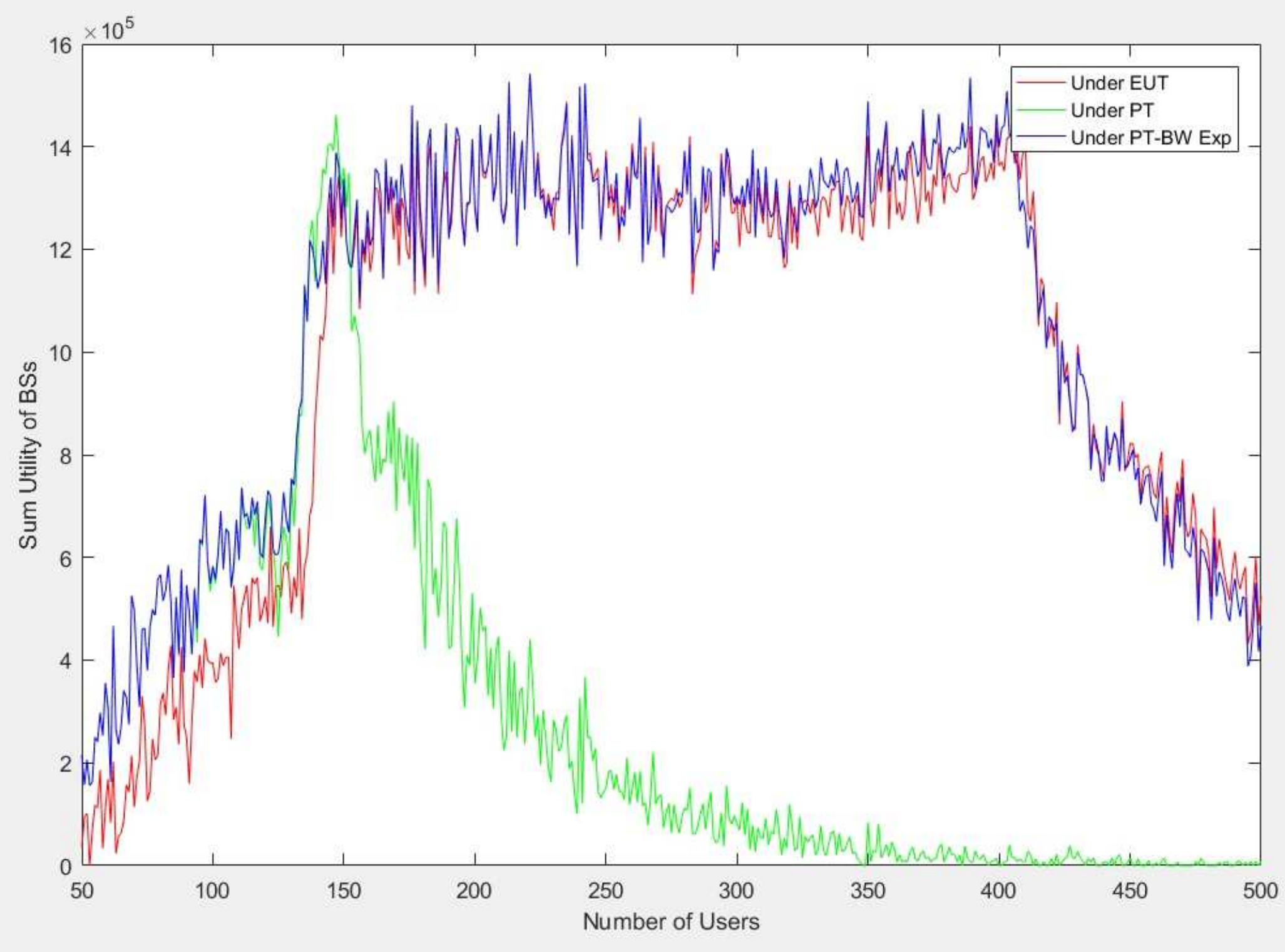}
  \centering
  \caption{Comparing sum utility of BSs under three different scenarios.}
  \label{fig8}
\end{figure}

When the number of users is less than 150, PT outperforms EUT in terms of SPs sum utility. The reason is that when the load is low, SPs offer higher data rates to users, hence, even with low service guarantees they can satisfy the user's minimum data rate constraint. And due to probability weighting effect (PWE) under PT, when the advertised service guarantees are lower than $1/e=0.37\%$, the user overestimates service guarantees, which leads to a higher acceptance rate for the SP offers under PT as compared to EUT. By increasing the acceptance rate, the SPs' sum utility is also increased. However, by increasing the number of users, the SPs advertised data rates will decrease as they assign less BW to each user. They have to increase their service guarantees to satisfy the user's data rate constraint. Hence, the user’s acceptance rate and consequently the SPs utilities decrease dramatically by increasing the number of users beyond $N=150$. However, using bandwidth expansion feature under PT, the SPs are able to retain most of their subscribed users with some extra cost. When the number of users goes beyond 400, which is the max network capacity under our setting, most of the SPs are not capable of satisfying user’s data rate constraint as their BW budget per user will diminish by increasing the load. Therefore, the number of users associated to BSs and consequently the SPs sum utility decreases when N is higher than 400. Fig. \ref{fig9} compares the sum utility of users under EUT, PT and PT with BW expansion for different load situations. As we expected, there are two turning points in this diagram. The first point is where the number of users goes beyond 150 users, which results in a considerable dropping of user association rate and the sum utility of users under PT. While by enabling bandwidth expansion under PT as its shown in Fig. \ref{fig9}, sum utility of users will be increased by increasing the number of users as more users will be associated to SPs until the number of users get close to the network capacity which is 400 users in our setting. After that, the user association rate reduces again as bid selection constraints for many users cannot be met by SPs considering their limited resources.

\begin{figure}[tb!]
  \includegraphics[width=\linewidth]{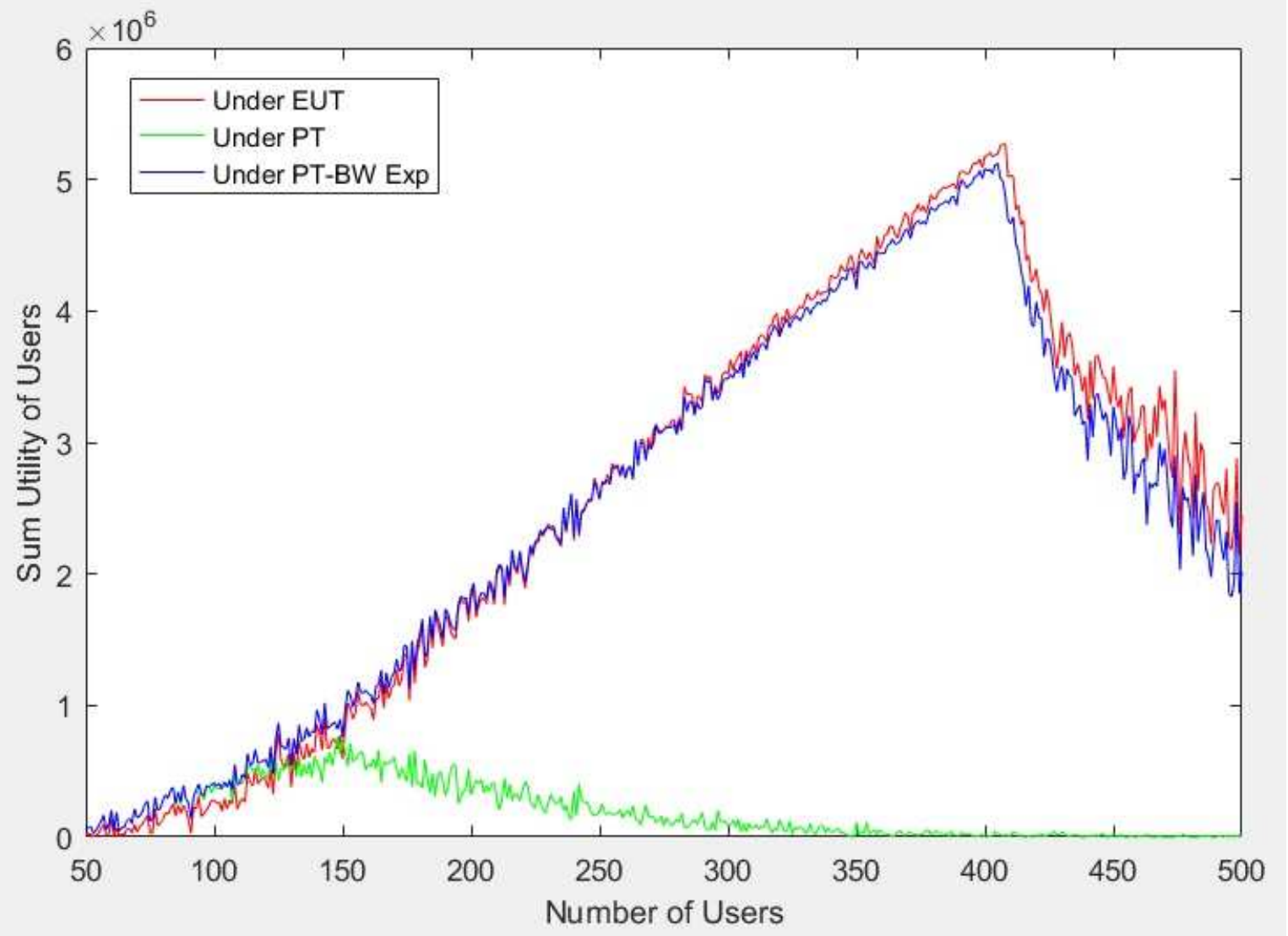}
  \centering
  \caption{Comparing sum utility of users under three different scenarios.}
  \label{fig9}
\end{figure}

To see how much extra costs the SPs have to incur to retain their users under PT using bandwidth expansion, we compare the average bandwidth consumption of users in EUT vs PT with bandwidth expansion feature in Fig. \ref{fig10}.
\begin{figure}[tb!]
  \includegraphics[width=\linewidth]{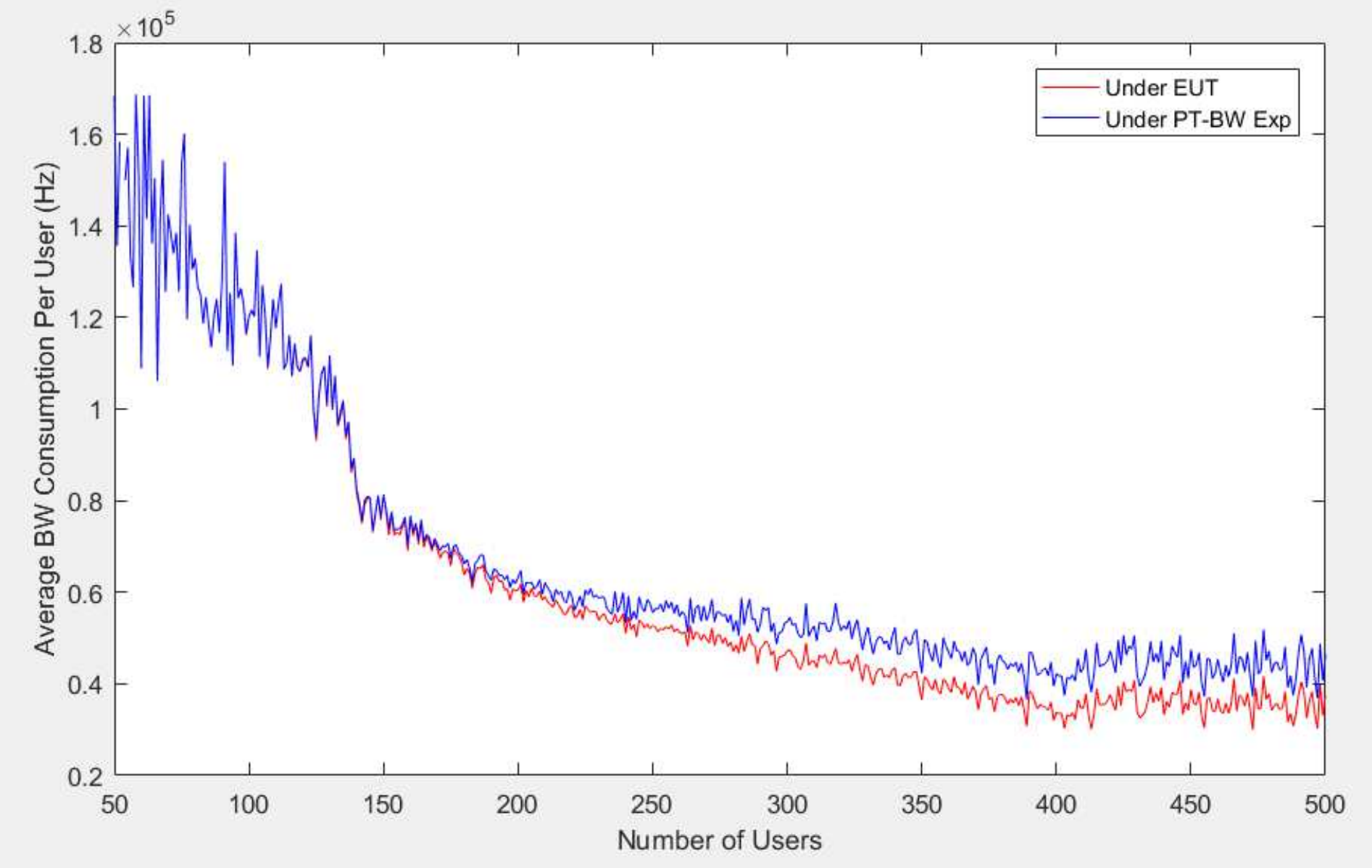}
  \centering
  \caption{Comparing Avg. BW consumption per user under EUT vs PT.}
  \label{fig10}
\end{figure}
When the SPs advertised service guarantees are less than $1/e=0.37\%$ threshold, which occurs when the number of users are less than 150, the SPs are not required to consume any extra BW to retain their EUT users under PT. However, by increasing the number of users beyond 150, the number of users who receive an offer with a service guarantee higher than $0.37\%$ will increase. So, to retain those users under PT, the SPs have to offer extra BW to satisfy user's bid selection constraints. The amount of extra BW that the SPs have to provide for the users under PT raises by increasing the load, as the SPs' service guarantees increase in this situation. So, user's underestimation of service guarantees become more intense. Note that we used the Perlec function with parameter $\alpha=0.7$ to capture PWE under PT. By using lower values for $\alpha$, the gap between these two curves become larger as the number of users is increasing. The simulation results under PT show that the SPs can retain most of their subscribed users under EUT by offering some extra bandwidth to them to compensate negative effects of the user's underestimation of their advertised service guarantees. Although bandwidth expansion increases the SPs cost, however as long as their expected payoff from user is higher than their cost, its justifiable for them to perform bandwidth expansion to retain their users.

\section{Conclusion}
\label{sec:conclusion}
This paper studied the problem of user association in wireless HetNets under PT, where all covering WiFi and cellular SPs offer data services to users who are free to accept or reject any of the received offers. We modeled this problem using a Stackelberg game, and extracted all potential NEs for this game under both EUT and PT. The NE existence analysis reveals that some of the feasible NEs under EUT become infeasible under PT when users underestimate the advertised service guarantees by the SPs. In this situation, the underestimation of service guarantees by the user increases the chance of rejection for the SPs bids, and hence causes the SP's average utility to diminish under PT. To avoid such a utility loss for the SPs under PT, we proposed a new bidding strategy by which the SPs are able to cope with the underestimation of their service guarantees by the user under PT. Our simulation results demonstrate that using such a bidding strategy, the SPs are able to retain most of their lost EUT users under PT by incurring some additional costs, and hence reduce their utility loss under PT.







\end{document}